\documentclass[12pt,a4paper]{article}

\usepackage{amsmath,amssymb}
\usepackage{graphicx}
\usepackage[nosort]{cite}

\usepackage{tikz}
\usetikzlibrary{decorations.markings}
\usepackage{epsfig}  
\usepackage{enumerate}
\usepackage{relsize}
\usepackage{dsfont}
\newcommand{\mathsym}[1]{{}}

\usepackage{amsthm,amsbsy,color,multicol}
\usepackage{array,calc}
\usepackage{rotating}
\usepackage{a4wide,bm}
\usepackage{bbm}
\usepackage[a4paper,text={450pt,650pt},centering]{geometry}
\usepackage{setspace}
\let\oldbfseries=\bfseries
\let\oldmdseries=\mdseries
\let\oldnormalfont=\normalfont
\renewcommand{\bfseries}{\oldbfseries\boldmath}
\renewcommand{\mdseries}{\oldmdseries\unboldmath}
\renewcommand{\normalfont}{\oldnormalfont\unboldmath}

\allowdisplaybreaks[3]

\numberwithin{equation}{section}

\usepackage[font=small,labelfont=bf,width=0.85\textwidth]{caption}

\ifx\hypersetup\sadfkjashdfkxja\newcommand\hypersetup[1]{}\fi

\hypersetup{plainpages=false}
\hypersetup{pdfpagemode=UseNone}
\hypersetup{bookmarksnumbered=true}
\hypersetup{pdfstartview=FitH}
\hypersetup{colorlinks=false}
\hypersetup{citebordercolor={.5 1 .5}}
\hypersetup{urlbordercolor={.5 1 1}}
\hypersetup{linkbordercolor={1 .7 .7}}
\DeclareMathSymbol{\Gamma}{\mathalpha}{letters}{"00}
\DeclareMathSymbol{\Delta}{\mathalpha}{letters}{"01}
\DeclareMathSymbol{\Theta}{\mathalpha}{letters}{"02}
\DeclareMathSymbol{\Lambda}{\mathalpha}{letters}{"03}
\DeclareMathSymbol{\Xi}{\mathalpha}{letters}{"04}
\DeclareMathSymbol{\Pi}{\mathalpha}{letters}{"05}
\DeclareMathSymbol{\Sigma}{\mathalpha}{letters}{"06}
\DeclareMathSymbol{\Upsilon}{\mathalpha}{letters}{"07}
\DeclareMathSymbol{\Phi}{\mathalpha}{letters}{"08}
\DeclareMathSymbol{\Psi}{\mathalpha}{letters}{"09}
\DeclareMathSymbol{\Omega}{\mathalpha}{letters}{"0A}
\DeclareMathOperator\arctanh{arctanh}

\newcommand{\gen}[1]{\mathrm{#1}}

\newcommand{\ii}{\mathrm{i}}

\newcommand*\widebar[1]{%
  \hbox{%
    \vbox{%
      \hrule height 0.5pt % The actual bar
      \kern0.25ex%         % Distance between bar and symbol
      \hbox{%
        \kern-0.3em%      % Shortening on the left side
        \ensuremath{#1}%
        \kern-0.1em%      % Shortening on the right side
      }%
    }%
  }%
}

\ifx\genfrac\sdflkaj\else\fi

\newcommand{\ket}[1]{\left|#1\right\rangle}      % Ket-Zustand
\newcommand{\bra}[1]{\left\langle #1\right|}     % Bra-Zustand
 
\newcommand{\prodop}[1]{\mathop{\overrightarrow\prod}\limits_{#1}}

\newcommand{\alg}[1]{\mathfrak{#1}}

\newcommand{\beq}{\begin{equation}}
\newcommand{\eeq}{\end{equation}}

\def\[{\begin{equation}}
\def\]{\end{equation}}
\def\<{\begin{eqnarray}}
\def\>{\end{eqnarray}}

\newtheorem{mydef}{Definition}

\newtheorem{theorem}{Theorem}
\newtheorem{lemma}{Lemma} 
\newtheorem{remark}{Remark}
\newtheorem{proposition}{Proposition}

\makeatletter
\def\mr@ignsp#1 {\ifx\:#1\@empty\else #1\expandafter\mr@ignsp\fi}%
\newcommand{\multiref}[1]{\begingroup%\let\protect\string%
\xdef\mr@no@sparg{\expandafter\mr@ignsp#1 \: }%
\def\mr@comma{}%
\@for\mr@refs:=\mr@no@sparg\do{\mr@comma\def\mr@comma{,}\ref{\mr@refs}}%
\endgroup}
\makeatother

\newcommand{\hypref}[2]{\ifx\href\asklfhas #2\else\href{#1}{#2}\fi}
\newcommand{\Secref}[1]{Section~\multiref{#1}}

\newcommand{\Appref}[1]{Appendix~\multiref{#1}}

\renewcommand{\eqref}[1]{(\multiref{#1})}

\makeatletter
\newlength{\apb@width}
\newcommand{\autoparbox}[2][c]{\settowidth{\apb@width}{#2}\parbox[#1]{\apb@width}{#2}}

\makeatother

\ifx\href\asklfhas\newcommand{\href}[2]{#2}\fi

\begin{document}

\renewcommand{\thefootnote}{\fnsymbol{footnote}}
%\pagenumbering{roman}
\thispagestyle{empty}
\begin{flushright}\footnotesize
ITP-UU-14/07 \\
SPIN-14/08
\end{flushright}
\vspace{1cm}

\begin{center}%
{\Large\bfseries%
\hypersetup{pdftitle={Partial differential equations from integrable vertex models}}%
Partial differential equations \\ from integrable vertex models%
\par} \vspace{2cm}%

\textsc{W. Galleas}\vspace{5mm}%
\hypersetup{pdfauthor={Wellington Galleas}}%

\textit{II. Institut f\"ur Theoretische Physik \\ Universit\"at Hamburg, Luruper Chaussee 149 \\ 22761 Hamburg, Germany}\vspace{3mm}% 

\verb+wellington.galleas@desy.de+ %

\par\vspace{3cm}

\textbf{Abstract}\vspace{7mm}

\begin{minipage}{12.7cm}
In this work we propose a mechanism for converting the spectral problem of vertex models transfer matrices
into the solution of certain linear partial differential equations. This mechanism is illustrated for the 
$U_q[\widehat{\alg{sl}}(2)]$ invariant six-vertex model and the resulting partial differential equation
is studied for particular values of the lattice length.

\hypersetup{pdfkeywords={Yang-Baxter algebra, vertex models, differential equations}}%
\hypersetup{pdfsubject={}}%
%\textbf{MSC Classes:}
%17B62, %Lie bialgebras; Lie coalgebras
%81U15, %Exactly and quasi-solvable systems
%17B67, %Kac-Moody (super)algebras; extended affine Lie algebras; toroidal Lie algebras
%17B25. %Exceptional (super)algebras
\end{minipage}
\vskip 2cm
{\small January 2015}

\end{center}

%%%%%%%%%%%%%%%%%%%%%%%%%%%%%%%%%%%%%%%%%%%%%%%%%%%%%%%%%%%%%%%%%%%%%%%%%%%%%%%%
\newpage
%\pagenumbering{arabic}	
%\setcounter{page}{1}
\renewcommand{\thefootnote}{\arabic{footnote}}
\setcounter{footnote}{0}

\tableofcontents

%%%%%%%%%%%%%%%%%%%%%%%%%%%%%%%%%%%%%%%%%%%%%%%%%%%%%%%%%%%%%%%%%%%%%%%%%%%%%%%%
\section{Introduction}
\label{sec:intro}

Vertex models of Statistical Mechanics are prominent examples where the computation of
the model partition function can be described as an eigenvalue problem. This possibility
dates back to Kramers and Wannier transfer matrix technique \cite{Kramers_1941a, Kramers_1941b}
originally devised for the Ising model. Within that approach the partition function of the model
is expressed in terms of the eigenvalues of a matrix usually refereed to as {\it transfer matrix}.
This technique has been successfully applied to a large variety of models although one has no guarantee 
a priori that the diagonalisation of the transfer matrix can be achieved. 

An important class of non-trivial models whose transfer matrix has been exactly diagonalized 
is formed by those possessing a parameter $\lambda$ such that its transfer matrix 
$T(\lambda)$ satisfies the commutativity condition $[T(\lambda_1), T(\lambda_2)]=0$ for general
values of $\lambda_1$ and $\lambda_2$ \cite{Baxter_book}. This latter property paves the way for
a variety of non-perturbative methods such as the coordinate Bethe ansatz \cite{Lieb_1967},
the algebraic Bethe ansatz \cite{Fad_1979}, $T\;$-$\; Q$ relations \cite{Baxter_1972}, 
analytical Bethe ansatz \cite{Reshet_1987} and separation of variables \cite{Sklyanin_1985a, Sklyanin_1985b}
among others.

Within the approach of the coordinate Bethe ansatz for instance, the transfer matrices 
diagonalisation process involves the explicit computation of the action of the transfer matrix on a finite
dimensional vector in terms of its components. However, suppose one would like to study the spectral
problem for an operator constituted by generators of the  $\alg{sl}(2)$ algebra. In that case we
could also consider a differential representation of the $\alg{sl}(2)$ algebra \cite{Humphreys_book} 
and study the same eigenvalue problem through the corresponding differential equation. 
It is worth remarking that we use this same methodology in Quantum Mechanics when we convert the spectral
problem for a Hamiltonian in the Heisenberg formulation into the solution of a stationary Schr\"odinger equation.

In this work we devise an analogous approach for the transfer matrix of the $U_q[\widehat{\alg{sl}}(2)]$
six-vertex model. More precisely, we obtain a partial differential equation describing the spectral problem
of the aforementioned operator. The main ingredient of our derivation is the Yang-Baxter algebra
employed within the lines of the algebraic-functional approach introduced in \cite{Galleas_2008} and
subsequently refined in the series of papers \cite{Galleas_2010, Galleas_2012, Galleas_SCP, Galleas_proc}.
It is also worth remarking that a connection between the spectral problem of transfer matrices and differential
equations had appeared previously in the literature under the name ODE/IM correspondence. See for instance the 
review \cite{Dorey_2007} and references therein. However, it is not clear if there is any relation between the ODE/IM correspondence 
and the approach considered here. The main reason for that lies in the fact that the ODE/IM correspondence describes
a relation between ordinary differential equations and integrable models while here we shall obtain partial differential equations.  
Moreover, our results are valid for finite lattices while the ODE/IM correspondence emerges in the continuum limit.

This paper is organized as follows. In \Secref{sec:eigen} we describe the $U_q[\widehat{\alg{sl}}(2)]$
transfer matrix and its spectral problem. \Secref{sec:AF} is devoted to the analysis of the aforementioned
spectral problem in terms of a functional equation derived as a direct consequence of the Yang-Baxter algebra.
This functional equation is converted into a partial differential equation in \Secref{sec:PDE} and its analysis
is performed for particular values of the lattice length. Concluding remarks are presented in \Secref{sec:conclusion}
and extra results for the case of domain wall boundaries are given in \Appref{sec:dwbc}.

%%%%%%%%%%%%%%%%%%%%%%%%%%%%%%%%%%%%%%%%%%%%%%%%%%%%%%%%%%%%%%%%%%%%%%%%%%%%%%%%
\section{The transfer matrix spectral problem}
\label{sec:eigen}

In this section we shall briefly recall some standard definitions and introduce a convenient notation
to describe the eigenvalue problem for the transfer matrix associated with the $U_q[\widehat{\alg{sl}}(2)]$
solution of the Yang-Baxter equation. Although this case corresponds to the well known trigonometric six-vertex model,
here we shall consider it from the perspective described in \cite{Galleas_proc}.

\paragraph{Monodromy and transfer matrices.} Let $\mathcal{T} \in \mbox{End}(\mathbb{V}_{\mathcal{A}} \otimes \mathbb{V}_{\mathcal{Q}})$
be an operator which we shall refer to as monodromy matrix. Here we shall consider the $U_q[\widehat{\alg{sl}}(2)]$ vertex model
and in that case we have $\mathbb{V}_{\mathcal{A}} \cong \mathbb{C}^2$ and $\mathbb{V}_{\mathcal{Q}} \cong (\mathbb{C}^2)^{\otimes L}$
for $L \in \mathbb{N}$. Hence the monodromy matrix $\mathcal{T}$ can be recasted as 
\[
\label{abcd}
\mathcal{T} = \left( \begin{matrix}
A & B \\
C & D \end{matrix} \right) \;
\]
with entries $A, B, C, D \in \mathbb{V}_{\mathcal{Q}}$. We then define the transfer matrix $T$ as the following
operator,
\[
\label{trans}
T = A + D \; .
\]

% \smallskip
\paragraph{Yang-Baxter algebra.} Integrable vertex models in the sense of Baxter are characterized by a monodromy matrix $\mathcal{T}$
satisfying the following algebraic relation,
\[
\label{ybalg}
\mathcal{R}(x-y) \left[ \mathcal{T}(x) \otimes \mathcal{T}(y) \right] = \left[ \mathcal{T}(y) \otimes \mathcal{T}(x) \right] \mathcal{R}(x-y) \; .
\]
In (\ref{ybalg}) we have spectral parameters $x, y \in \mathbb{C}$ and $\mathcal{R} \in \mbox{End}( \mathbb{C}^2 \otimes \mathbb{C}^2 )$
in the case of the $U_q[\widehat{\alg{sl}}(2)]$ vertex model. The algebraic relation (\ref{ybalg}) is associative for 
$\mathcal{R}$-matrices satisfying the Yang-Baxter equation, namely
\[
\label{ybe}
\left[ \mathcal{R}(x) \otimes \mathbbm{1} \right] \left[ \mathbbm{1} \otimes \mathcal{R}(x+y) \right] \left[ \mathcal{R}(y) \otimes \mathbbm{1} \right] = \left[ \mathbbm{1} \otimes \mathcal{R}(y) \right]   \left[ \mathcal{R}(x+y) \otimes \mathbbm{1} \right]  \left[ \mathbbm{1} \otimes \mathcal{R}(x) \right] \; ,
\]
with symbol $\mathbbm{1}$ denoting the $2 \times 2$ identity matrix. The matrix $\mathcal{R}$ plays the role of structure constant for the
Yang-Baxter algebra (\ref{ybalg}) and the solution of (\ref{ybe}) invariant under the $U_q[\widehat{\alg{sl}}(2)]$ algebra explicitly
reads
\[
\label{rmat}
\mathcal{R}(x) = \left( \begin{matrix}
a(x) & 0 & 0 & 0 \\
0 & c(x) & b(x) & 0 \\
0 & b(x) & c(x) & 0 \\
0 & 0 & 0 & a(x) \end{matrix} \right) \; .
\]
The non-null entries of (\ref{rmat}) are given by functions $a(x) = \sinh{(x + \gamma)}$, $b(x) = \sinh{(x)}$
and $c(x) = \sinh{(\gamma)}$. 

\begin{remark}
Since the $\mathcal{R}$-matrix (\ref{rmat}) is invertible, the relation (\ref{ybalg}) implies that 
the transfer matrix (\ref{trans}) forms a commutative family, i.e. $[T(x) , T(y)] = 0$.
\end{remark}
 
Now let $\mathcal{M}$ be the set $\mathcal{M}(x) = \{ A, B , C, D\}(x)$ parameterized by a continuous complex variable $x$. 
The elements of $\mathcal{M}$ obey the Yang-Baxter algebra and among the commutation rules encoded in 
(\ref{ybalg}) we shall make use of the following ones:
\<
\label{ABDB}
A(x_1) B(x_2) &=& \frac{a(x_2 - x_1)}{b(x_2 - x_1)} B(x_2) A(x_1) - \frac{c(x_2 - x_1)}{b(x_2 - x_1)} B(x_1) A(x_2) \nonumber \\
D(x_1) B(x_2) &=& \frac{a(x_1 - x_2)}{b(x_1 - x_2)} B(x_2) D(x_1) - \frac{c(x_1 - x_2)}{b(x_1 - x_2)} B(x_1) D(x_2) \nonumber \\
B(x_1) B(x_2) &=& B(x_2) B(x_1) \; .
\> 
\begin{remark}
The elements of $\mathcal{M}$ are subjected to the Yang-Baxter algebra (\ref{ybalg}) which is in general 
non-abelian. In this way the $2$-tuple $(\xi_1 , \xi_2)  : \; \xi_{i} \in \mathcal{M}(\lambda_i)$
originated from the Cartesian product $\mathcal{M}(\lambda_1) \times \mathcal{M}(\lambda_2)$ will be 
simply understood as the non-commutative product $\xi_1 \xi_2$. This convention is naturally extended
for the $n$-tuples $(\xi_1 , \dots , \xi_n)$ generated by the products $\mathcal{M}(\lambda_1) \times \dots \times \mathcal{M}(\lambda_n)$.
\end{remark}

% \smallskip
\paragraph{Monodromy matrix representation.} The ordered product 
\[
\label{mono}
\mathcal{T}(\lambda) = \prodop{1 \leq j \leq L} \gen{P}_{\mathcal{A} j} \mathcal{R}_{\mathcal{A} j}(\lambda - \mu_j)
\]
with $\mathcal{R}$-matrix given by (\ref{rmat}) is a representation of (\ref{ybalg}) due to the Yang-Baxter relation (\ref{ybe}). 
In (\ref{mono}) $\gen{P}$ denotes the standard permutation matrix $\gen{P}_{l m}  : \mathbb{V}_l \otimes \mathbb{V}_m \mapsto \mathbb{V}_m \otimes \mathbb{V}_l$
for $\mathbb{V}_{l,m} \cong \mathbb{C}^2$, while $\lambda, \mu_j \in \mathbb{C}$ are respectively the spectral 
and inhomogeneity parameters. In its turn the subscripts in $\mathcal{R}_{\mathcal{A} j}$ indicate that we have a $\mathcal{R}$-matrix
acting on the tensor product space $\mathbb{V}_{\mathcal{A}} \otimes \mathbb{V}_{j}$. More precisely, we have
$\mathcal{R}_{\mathcal{A} j} \in \mbox{End}(\mathbb{V}_{\mathcal{A}} \otimes \mathbb{V}_{j})$.

% \smallskip
\paragraph{Highest weight vectors.} The vector $\ket{0} = \left( \begin{matrix} 1 \\ 0  \end{matrix}  \right)^{\otimes L}$
is a $\alg{sl}(2)$ highest weight vector and the action of $\mathcal{M}$ built from the representation (\ref{mono}) can be
straightforwardly computed due to the structure of (\ref{rmat}). They are given by the following expressions:
\begin{align}
\label{action}
A(\lambda) \ket{0} &= \prod_{j=1}^L a(\lambda - \mu_j) \ket{0} & B(\lambda) \ket{0} &\neq 0 \nonumber \\
C(\lambda) \ket{0} &= 0 &  D(\lambda) \ket{0} &= \prod_{j=1}^L b(\lambda - \mu_j) \ket{0} \;\;\; .
\end{align}

% \smallskip
\paragraph{Yang-Baxter relations of higher degrees.} The spectrum of the transfer matrix $T$ can be encoded 
in a set of functional relations along the lines described in \cite{Galleas_2008, Galleas_proc, Galleas_2013}. 
Those functional relations are derived as a direct consequence of the Yang-Baxter algebra.
In order to simplify our notation we introduce the symbol $[ \lambda_1 , \dots , \lambda_n ]$ defined as the
following product of operators,
\[
\label{prodB}
[ \lambda_1 , \dots , \lambda_n ] = \prodop{1 \leq j \leq n} B(\lambda_j) \; .
\]

\begin{remark} 
Due to the last relation of (\ref{ABDB}) we have that $[\lambda_1 , \dots , \lambda_n]$ is symmetric under
the permutation of variables, i.e. $[\dots, \lambda_i , \dots , \lambda_j, \dots] = [\dots, \lambda_j , \dots , \lambda_i, \dots]$.
This property motivates the set theoretic notation $X^{a,b} = \{ \lambda_j \; : \; a \leq j \leq b  \}$ 
and we can write $[ \lambda_1 , \dots , \lambda_n ]  = [ X^{1,n} ]$. We shall also employ the notation $X^{a,b}_{\lambda} = X^{a,b} \backslash \{ \lambda \}$.
\end{remark}

Now the products $A(\lambda_0) [ X^{1,n} ]$ and $D(\lambda_0) [ X^{1,n} ]$ can be investigated under the light of (\ref{ABDB}) taking
into account the previous definitions. By doing so we are left with the following Yang-Baxter relations of degree $n+1$,
\<
\label{off}
A(\lambda_0) [ X^{1,n} ] &=& \prod_{\lambda \in X^{1,n}} \frac{a(\lambda - \lambda_0)}{b(\lambda - \lambda_0)} [X^{1,n}] A(\lambda_0) \nonumber \\
&& - \; \sum_{\lambda \in X^{1,n}} \frac{c(\lambda - \lambda_0)}{b(\lambda - \lambda_0)} \prod_{\tilde{\lambda} \in X^{1,n}_{\lambda}} \frac{a(\tilde{\lambda} - \lambda)}{b(\tilde{\lambda} - \lambda)} [X^{0,n}_{\lambda}] A(\lambda) \nonumber \\
D(\lambda_0) [ X^{1,n} ] &=& \prod_{\lambda \in X^{1,n}} \frac{a(\lambda_0 - \lambda)}{b(\lambda_0 - \lambda)} [X^{1,n}] D(\lambda_0) \nonumber \\
&& - \; \sum_{\lambda \in X^{1,n}} \frac{c(\lambda_0 - \lambda)}{b(\lambda_0 - \lambda)} \prod_{\tilde{\lambda} \in X^{1,n}_{\lambda}} \frac{a(\lambda - \tilde{\lambda})}{b(\lambda - \tilde{\lambda})} [X^{0,n}_{\lambda}] D(\lambda) \; .
\>
Here we intend to explore the relations (\ref{off}) in order to describe the spectrum of the transfer matrix
(\ref{trans}). With that goal in mind we add up both expressions in (\ref{off}) to obtain the identity
\<
\label{Toff}
T(\lambda_0) [ X^{1,n} ] &=&  [X^{1,n}] ( M^{A}_0 A(\lambda_0) + M^{D}_0 D(\lambda_0) ) \nonumber \\
&& - \; \sum_{\lambda \in X^{1,n}}  [X^{0,n}_{\lambda}] ( M^{A}_{\lambda} A(\lambda) + M^{D}_{\lambda} D(\lambda) ) \; .
\>
The coefficients in (\ref{Toff}) explicitly read
\begin{align}
\label{coeff}
M^{A}_0 &= \prod_{\lambda \in X^{1,n}} \frac{a(\lambda - \lambda_0)}{b(\lambda - \lambda_0)} & M^{D}_0 &= \prod_{\lambda \in X^{1,n}} \frac{a(\lambda_0 - \lambda)}{b(\lambda_0 - \lambda)} \nonumber \\
M^{A}_{\lambda} &= \frac{c(\lambda - \lambda_0)}{b(\lambda - \lambda_0)} \prod_{\tilde{\lambda} \in X^{1,n}_{\lambda}} \frac{a(\tilde{\lambda} - \lambda)}{b(\tilde{\lambda} - \lambda)} & M^{D}_{\lambda} &= \frac{c(\lambda_0 - \lambda)}{b(\lambda_0 - \lambda)} \prod_{\tilde{\lambda} \in X^{1,n}_{\lambda}} \frac{a(\lambda - \tilde{\lambda})}{b(\lambda - \tilde{\lambda})} \; .
\end{align}

Within the framework of the algebraic Bethe ansatz \cite{Fad_1979} one would then consider the action of 
(\ref{Toff}) on the highest weight vector $\ket{0}$ and try to fix the set of parameters 
$X^{1,n}$ in such a way that the transfer matrix eigenvalues can be directly read from the resulting
expression. Here we follow a different strategy and we shall demonstrate how the Yang-Baxter relation 
(\ref{Toff}) can be converted into a functional equation.

%%%%%%%%%%%%%%%%%%%%%%%%%%%%%%%%%%%%%%%%%%%%%%%%%%%%%%%%%%%%%%%%%%%%%%%%%%%%%%%%
\section{Algebraic-functional approach}
\label{sec:AF}

In this section we aim to show that the eigenvalue problem for the transfer matrix (\ref{trans})
can be described in terms of certain functional equations originated from the Yang-Baxter
algebra. This statement is made precise in Theorem \ref{funeq} and its proof will require the 
following definitions.

\begin{mydef}[Algebras into functions] Considering the mechanism described in \cite{Galleas_proc} we 
define the following continuous and additive map $\pi$,
\[
\label{pi}
\pi_{n+1} \; : \quad \mathcal{M}(\lambda_0) \times \mathcal{M}(\lambda_1) \times \dots \times \mathcal{M}(\lambda_n) \mapsto \mathbb{C} [\lambda_0^{\pm 1} , \lambda_1^{\pm 1} , \dots , \lambda_n^{\pm 1} ]  \; .
\]
The map $\pi_{n+1}$ essentially associates a complex function to the elements of $\mathcal{M}(\lambda_0) \times \mathcal{M}(\lambda_1) \times \dots \times \mathcal{M}(\lambda_n)$.
\end{mydef}
The proof of the announced Theorem \ref{funeq} will require the application of the map $\pi_{n+1}$ over the higher order Yang-Baxter relation 
(\ref{Toff}). We shall also need to build a suitable realization of (\ref{pi}) and here we consider the recipe given in \cite{Galleas_2013}.

\begin{lemma} \label{pion}
Let $\ket{ \gen{\Lambda} } \in \mathrm{span}(\mathbb{V}_{\mathcal{Q}})$ be an eigenvector of the transfer matrix (\ref{trans})
and let $\bra{ \gen{\Lambda} } $ denote its dual. Also consider the $\alg{sl}(2)$ highest weight vector $\ket{0}$ as previously
defined and $\mathcal{W}_{n+1} = \mathcal{M}(\lambda_0) \times \mathcal{M}(\lambda_1) \times \dots \times \mathcal{M}(\lambda_n)$.
Hence we have that
\[
\label{pirep}
\pi_{n+1} (\mathcal{A}) = \bra{ \gen{\Lambda} } \mathcal{A} \ket{0} \qquad \qquad \forall \; \mathcal{A} \in \mathcal{W}_{n+1} 
\]
is a realization of (\ref{pi}).
\end{lemma}
\begin{proof}
The proof is straightforward and follows from the fact that both $\ket{0}$ and $\ket{ \gen{\Lambda} }$ do not depend on
the variables $\lambda_j$. The independence of $\ket{0}$ with $\lambda_j$ is clear from its definition while
this same property for $\ket{ \gen{\Lambda} }$ is due to the fact that the transfer matrix (\ref{trans}) forms
a commutative family.
\end{proof}

\begin{theorem}[Functional equation] \label{funeq}
Let $\Lambda$ be an eigenvalue of the transfer matrix $T$ associated with the eigenvector 
$\ket{ \gen{\Lambda} } \in \mathrm{span}(\mathbb{V}_{\mathcal{Q}})$, i.e. $T(\lambda) \ket{\gen{\Lambda}} = \Lambda(\lambda) \ket{\gen{\Lambda}}$.  
Then $\exists \; \mathcal{F}_n \; : \; \mathbb{C}^n \mapsto \mathbb{C}$ characterizing the eigenvalue $\Lambda$ through the equation
\[
\label{FZ}
J_0 \mathcal{F}_n (X^{1,n}) - \sum_{\lambda \in X^{1,n}} K_{\lambda} \mathcal{F}_n (X^{0,n}_{\lambda}) = \Lambda(\lambda_0) \mathcal{F}_n (X^{1,n}) 
\]
with coefficients
\<
\label{coeffJK}
J_0 &=& \prod_{j=1}^L a(\lambda_0 - \mu_j) \; M^{A}_0 + \prod_{j=1}^L b(\lambda_0 - \mu_j) \; M^{D}_0 \nonumber \\
K_{\lambda} &=& \prod_{j=1}^L a(\lambda - \mu_j) \; M^{A}_{\lambda} + \prod_{j=1}^L b(\lambda - \mu_j) \; M^{D}_{\lambda} \; .
\>
\end{theorem}
\begin{proof}
The realization (\ref{pirep}) exhibits useful properties which will aid us in extracting information
about the transfer matrix spectrum from the higher degree relation (\ref{Toff}). For instance, the
application of $\pi_{n+1}$ to the LHS of (\ref{Toff}) will produce the term
$\pi_{n+1} (T(\lambda_0) [ X^{1,n} ])$ which simplifies to
\[
\label{lhs}
\pi_{n+1} (T(\lambda_0) [ X^{1,n} ]) = \Lambda(\lambda_0) \pi_{n} ([ X^{1,n} ]) \; .
\]
Similarly, we find that the application of the map $\pi_{n+1}$ over the RHS of (\ref{Toff}) only yields terms of the form $\pi_{n+1} ([ Z^{1,n} ] A(x) )$ and
$\pi_{n+1} ([ Z^{1,n} ] D(x) )$ for generic variables $x$ and $Z^{1,n} = \{ z_j \in \mathbb{C} \; : \; 1 \leq j \leq n  \}$.
Due to (\ref{action}) and (\ref{pirep}) those terms exhibit the following reduction properties,
\<
\label{rhs}
\pi_{n+1} ([ Z^{1,n} ] A(x) ) &=& \prod_{j=1}^{L} a(x-\mu_j) \pi_{n} ([ Z^{1,n} ] ) \nonumber \\
\pi_{n+1} ([ Z^{1,n} ] D(x) ) &=& \prod_{j=1}^{L} b(x-\mu_j) \pi_{n} ([ Z^{1,n} ] ) \; .
\>
In their turn the relations (\ref{lhs}) and (\ref{rhs}) tell us that the map $\pi$ given by (\ref{pirep})
obeys recurrence relations of type $\pi_{n+1} \mapsto \pi_n$ over the elements of (\ref{Toff}).
Next we introduce the notation $\mathcal{F}_n( X^{1,n} ) = \pi_n ( [X^{1,n}] )$ in such a way that
the application of (\ref{pirep}) on (\ref{Toff}), taking into account the properties (\ref{lhs}) and (\ref{rhs}), 
yields the functional relation
\<
J_0 \mathcal{F}_n (X^{1,n}) - \sum_{\lambda \in X^{1,n}} K_{\lambda} \mathcal{F}_n (X^{0,n}_{\lambda}) = \Lambda(\lambda_0) \mathcal{F}_n (X^{1,n}) \; ,
\>
with coefficients
\<
\label{coeffJK}
J_0 &=& \prod_{j=1}^L a(\lambda_0 - \mu_j) \; M^{A}_0 + \prod_{j=1}^L b(\lambda_0 - \mu_j) \; M^{D}_0 \nonumber \\
K_{\lambda} &=& \prod_{j=1}^L a(\lambda - \mu_j) \; M^{A}_{\lambda} + \prod_{j=1}^L b(\lambda - \mu_j) \; M^{D}_{\lambda} \; .
\>
This completes the proof of Theorem \ref{funeq}.
\end{proof}

%%%%%%%%%%%%%%%%%%%%%%%%%%%%%%%%%%%%%%%%%%%%%%%%%%%%%%%%%%%%%%%%%%%%%%%%%%%%%%%%
\subsection{Operatorial description}
\label{sec:OPD}

The transfer matrix eigenvalue problem has been described in Theorem \ref{funeq} as the solution
of a functional equation. In this section we intend to show that the obtained functional equation,
namely (\ref{FZ}), can be recasted in an operatorial form which allows us to identify the action 
of the transfer matrix (\ref{trans}) on a particular function space. For that we introduce the
operator $D_{z_i}^{z_{\alpha}}$ whose properties are described as follows.

\begin{mydef}
Let $n \in \mathbb{N}$ and $\alpha \notin \{ k \in \mathbb{N} \; : \; 1 \leq k \leq n \}$.
Also, let $f$ be a complex function $f(z) \in \mathbb{C}[z]$ where $z=(z_1 , \dots , z_n) \in \mathbb{C}^n$.
Then we define the action of the operator $D_{z_i}^{z_{\alpha}}$ on $\mathbb{C}[z]$ as,
\[
\label{dia}
D_{z_i}^{z_{\alpha}} \; : \qquad \quad f(z_1, \dots , z_i , \dots, z_n) \;\; \mapsto \;\; f(z_1, \dots , z_{\alpha} , \dots, z_n) \; .
\]
The operator $D_{z_i}^{z_{\alpha}}$ has been previously introduced in \cite{Galleas_2011} and it basically replaces a given
variable $z_i$ by a variable $z_{\alpha}$. 
\end{mydef}

In terms of the operator $D_{z_i}^{z_{\alpha}}$ we can rewrite Eq. (\ref{FZ}) as 
\[
\label{eigenL}
\mathfrak{L}(\lambda_0) \mathcal{F}_n (X^{1,n}) = \Lambda(\lambda_0) \mathcal{F}_n (X^{1,n})
\]
with operator $\mathfrak{L}$ reading
\[
\label{Lop}
\mathfrak{L}(\lambda_0) = J_0 - \sum_{\lambda \in X^{1,n}} K_{\lambda} D^{\lambda_0}_{\lambda} \; .
\]
Now we can immediately recognize Eq. (\ref{eigenL}) as an eigenvalue equation and some comments are
in order at this stage. For instance, the introduction of the operator $D_{z_i}^{z_{\alpha}}$ is able to 
localize the whole dependence of the LHS of (\ref{FZ}) with the spectral parameter $\lambda_0$
in the operator $\mathfrak{L}$. In fact, we can identify $\mathfrak{L}$ with the transfer matrix
(\ref{trans}) acting on a particular function space spanned by functions $\mathcal{F}_n$. This 
function space will be described in the next section.

%%%%%%%%%%%%%%%%%%%%%%%%%%%%%%%%%%%%%%%%%%%%%%%%%%%%%%%%%%%%%%%%%%%%%%%%%%%%%%%%
\subsection{The function space $\gen{\Xi}( \mathbb{C}^n )$}
\label{sec:Fn}

The functions $\mathcal{F}_n$ solving Eq. (\ref{FZ}) consist of the projection of the dual
transfer matrix eigenvector onto a particular set of vectors usually refereed to as Bethe vectors.
Although we are mainly interested in the eigenvalues $\Lambda$ we still need to 
restrict our solutions $\mathcal{F}_n$ to a class of functions preserving certain representation
theoretic properties exhibited by the elements involved in the derivation of (\ref{FZ}).
This fact motivates the definition of the function space $\gen{\Xi}( \mathbb{C}^n )$ whose properties
will be discussed in what follows.

\begin{mydef} Let the functions $\mathcal{F}_n \; : \;\; \mathbb{C}^n \mapsto \mathbb{C}$ be of the form
\[
\label{fn}
\mathcal{F}_n (X^{1,n}) = \bra{\gen{\Lambda}} [\lambda_1 , \dots , \lambda_n]  \ket{0} \; ,
\]
where each operator $B$ in the product $[\lambda_1 , \dots , \lambda_n]$ is built according
to (\ref{abcd}), (\ref{rmat}) and (\ref{mono}). In its turn $\bra{\gen{\Lambda}}$ is a dual eigenvector
of the transfer matrix (\ref{trans}) whilst $\ket{0}$ is the $\alg{sl}(2)$ highest weight vector as previously defined. 
\end{mydef}

We can see from Lemma \ref{pion} that the whole dependence of $\mathcal{F}_n$ with a given variable
$\lambda_j$ comes from the operator $B(\lambda_j)$. Thus the characterization of $\gen{\Xi}( \mathbb{C}^n )$
can be performed with the help of the following Proposition.

\begin{proposition}[Polynomial structure] \label{polB}
The operator $B(\lambda_i)$ is of the form
\[
B(\lambda_i) = x_i^{\frac{1-L}{2}} P_B (x_i)
\]
where $x_i = e^{2 \lambda_i}$ and $P_B$ is a polynomial of degree $L-1$. 
\end{proposition}
\begin{proof}
The proof follows from induction and it can be found with details in \cite{Galleas_2008}.
\end{proof}

\begin{mydef} \label{Km}
Let $\mathbb{K} [x_1 , x_2 , \dots , x_n]$ be the polynomial ring in $n$ variables
$x_1, \dots , x_n$ which we shall simply denote as $\mathbb{K} [x]$. Then we define
$\mathbb{K}^m [x] \subset \mathbb{K} [x]$ as the subset of $\mathbb{K} [x]$ formed 
by  polynomials of degree $m$ in each variable $x_i$.
\end{mydef}

Due to the Proposition \ref{polB} we can conclude that $\mathcal{F}_n$ are of the form
\[
\label{Fbar}
\mathcal{F}_n (X^{1,n}) = \prod_{i=1}^{n} x_i^{\frac{1-L}{2}} \; \bar{\mathcal{F}}_n (x_1, x_2 , \dots , x_n) \; ,
\]
where $\bar{\mathcal{F}}_n$ is a multivariate polynomial of degree $L-1$ in each one of its variables.
Taking into account the Definition \ref{Km} we can write $x=(x_1, \dots , x_n) \in \mathbb{C}^n$ and conclude
that $\bar{\mathcal{F}}_n  = \bar{\mathcal{F}}_n (x) \in \mathbb{K}^{L-1} [x]$. The function space 
$\gen{\Xi}$ is then defined as follows.

\begin{mydef}[Space $\gen{\Xi}$] 
The function space $\gen{\Xi}(\mathbb{C}^n)$ consists of the following set of functions,
\<
\label{XI}
\gen{\Xi}(\mathbb{C}^n) = \left\{ \mathcal{F}_n \; : \; \mathfrak{L}(\lambda) \mathcal{F}_n = \Lambda(\lambda) \mathcal{F}_n , \; \mathcal{F}_n = \mathbf{x}^{\frac{1-L}{2}} \; \bar{\mathcal{F}}_n (x), \; \bar{\mathcal{F}}_n  \in \mathbb{K}^{L-1} [x] \right\} \; ,
\>
where $\mathbf{x}^{\frac{1-L}{2}} = \prod_{i=1}^n x_i^{\frac{1-L}{2}}$.
\end{mydef}

\begin{remark} 
We can readily see from Definition $4$ that $\gen{\Xi}(\mathbb{C}^n) \subset \mathbf{x}^{\frac{1-L}{2}} \mathbb{K}^{L-1} [x]$. 
\end{remark}

%%%%%%%%%%%%%%%%%%%%%%%%%%%%%%%%%%%%%%%%%%%%%%%%%%%%%%%%%%%%%%%%%%%%%%%%%%%%%%%%
\section{Partial differential equations}
\label{sec:PDE}

The operator $\mathfrak{L}$ defined in (\ref{Lop}) corresponds to the transfer matrix in the
function space $\gen{\Xi}(\mathbb{C}^n)$. In its turn $\mathfrak{L}$ is given in terms
of operators $D_{z_i}^{z_{\alpha}}$ and here we intend to demonstrate that those operators admit a 
differential realization when their action is restricted to the set $\mathbb{K}^m [x]$.
This realization could not be immediately employed for (\ref{Lop}) as we are interested in solutions 
$\mathcal{F}_n \in \gen{\Xi}(\mathbb{C}^n)$. Nevertheless, in what follows we shall see how this differential
structure can still be incorporated into Eq. (\ref{eigenL}). This approach has been previously employed in \cite{Galleas_proc},
where we have derived a set of partial differential equations (PDEs) satisfied by the partition function of the six-vertex model
with domain wall boundaries. For completeness' sake, in the \Appref{sec:dwbc} we also discuss one of the PDEs explicitly 
obtained in \cite{Galleas_proc}.

\begin{lemma}[Differential realization] \label{diff_real}
Let $\mathbb{K}^m [z] \subset \mathbb{K} [z]$ with $z=(z_1, \dots , z_n) \in \mathbb{C}^n$ be a subset of the polynomial
ring according to the Definition \ref{Km}. The operator $D_{z_i}^{z_{\alpha}}$ in $\mathbb{K}^m [z]$ is then given by
\[
\label{diarep}
D_{z_i}^{z_{\alpha}} = \sum_{k=0}^{m} \frac{(z_{\alpha} - z_i)^k}{k!} \frac{\partial^k}{\partial z_i^k} \; .
\]
\end{lemma}
\begin{proof} 
A detailed proof is given in \cite{Galleas_proc}.
\end{proof}

As we have previously remarked we can not immediately substitute the realization (\ref{diarep}) into (\ref{Lop})
as the functions $\mathcal{F}_n$ belong to the function space $\gen{\Xi}(\mathbb{C}^n)$.
However, as the non-polynomial part of $\gen{\Xi}(\mathbb{C}^n)$ consists of an overall multiplicative factor,
we can still rewrite (\ref{eigenL}) in terms of functions $\bar{\mathcal{F}}_n \in \mathbb{K}^{L-1}[x]$ defined through (\ref{Fbar}).
For that we introduce the variable $z=e^{2 \lambda}$ and define the functions
\<
\label{coffeebar}
\bar{J}_0 = J_0 x_0^{\frac{L}{2}} \; , \qquad   \bar{K}_{z} = K_{\lambda} x_0^{\frac{1}{2}} z^{\frac{L-1}{2}} \quad
\mbox{and} \quad \bar{\Lambda}(x_0) = \Lambda(\lambda_0) x_0^{\frac{L}{2}} \; .  
\>
By doing so we are left with the equation
\[
\label{eigen_bar}
\bar{\mathfrak{L}}(x_0) \bar{\mathcal{F}}_n (\bar{X}^{1,n}) = \bar{\Lambda}(x_0) \bar{\mathcal{F}}_n (\bar{X}^{1,n}) 
\]
where $\bar{X}^{a,b} = \{ x_k \; : \; a \leq k \leq b  \}$. In its turn the operator $\bar{\mathfrak{L}}$ reads
\[
\label{Lop_bar}
\bar{\mathfrak{L}}(x_0) = \bar{J}_0 - \sum_{x \in \bar{X}^{1,n}} \bar{K}_{x} D^{x_0}_{x} 
\]
and now it acts on functions $\bar{\mathcal{F}}_n \in \mathbb{K}^{L-1}[x]$. Hence we can employ the realization
(\ref{diarep}) and this procedure reveals that $\bar{\mathfrak{L}}$ is of the form 
\[
\label{Lop_bard}
\bar{\mathfrak{L}}(x_0) = \sum_{k=0}^{L} x_0^k \; \gen{\Omega}_k \; .
\]
Here $\{ \gen{\Omega}_k \}$ is a set of differential operators and the expression (\ref{Lop_bard}) implies 
that the LHS of (\ref{eigen_bar}) is a polynomial of degree $L$ in the variable $x_0$. 

\begin{remark}
The operator $\bar{\mathfrak{L}}$ corresponds to the transfer matrix $T$ in the function space $\mathbb{K}^{L-1}[x]$.
Thus, since $[T(\lambda_1) , T(\lambda_2)]=0$ as matricial operators, we can conclude that $[ \bar{\mathfrak{L}}(x_1) , \bar{\mathfrak{L}}(x_2) ] = 0$
which implies the condition $[ \gen{\Omega}_i ,  \gen{\Omega}_j ] = 0$.
\end{remark}

The RHS of (\ref{eigen_bar}) is also a polynomial of degree $L$ in the variable $x_0$ and this feature prevents that any operator
$\gen{\Omega}_k$ vanishes identically by construction. This property can be demonstrated with the help of Proposition \ref{propAD}.

\begin{proposition} \label{propAD}
The operators $A(\lambda)$ and $D(\lambda)$ are of the form 
\<
\label{ADP}
A(\lambda) = z^{-\frac{L}{2}} P_A (z) \qquad \mbox{and} \qquad D(\lambda) = z^{-\frac{L}{2}} P_D (z) \; ,
\>
where $P_A$ and $P_D$ are polynomials of degree $L$.
\end{proposition}
\begin{proof}
The proof follows from induction and the details can be found in \cite{Galleas_2008}.
\end{proof}

Now from equations (\ref{trans}) and (\ref{ADP}) we can conclude that $T(\lambda)= z^{-\frac{L}{2}} P_T (z)$
where $P_T$ is a polynomial of degree $L$. Moreover, since our transfer matrix forms a commutative
family we can conclude that its eigenvalues will be of the form $\Lambda(\lambda)= z^{-\frac{L}{2}} P_{\Lambda} (z)$
where $P_{\Lambda}$ is also a polynomial of degree $L$. Hence, the RHS of (\ref{eigen_bar}) 
consists of a polynomial of degree $L$ in the variable $x_0$ and we can write
\[
\bar{\Lambda}(x_0) = \sum_{k=0}^L x_0^k \; \Delta_k \; .
\]
In this way Eq. (\ref{eigen_bar}) must be satisfied independently by each power in $x_0$
and we are left with the following system of differential equations,
\[
\label{eigK}
\gen{\Omega}_k \bar{\mathcal{F}}_n (x_1, \dots , x_n) = \Delta_k \bar{\mathcal{F}}_n (x_1, \dots , x_n) \qquad \quad 0 \leq k \leq L \; .
\]
The system of Eqs. (\ref{eigK}) comprises a total of $L+1$ eigenvalue problems, i.e. an eigenvalue equation
for each operator $\gen{\Omega}_k$, being solved by the same eigenfunction $\bar{\mathcal{F}}_n$. Moreover, the direct
inspection of (\ref{eigK}) for small values of $n$ and $L$ shows that each equation is solely able to determine the eigenfunctions
$\bar{\mathcal{F}}_n$ in addition to its eigenvalue $\Delta_k$. Thus the system of differential equations (\ref{eigK})
can be simultaneously integrated. 

The explicit form of the operators $\gen{\Omega}_k$ can be straightforwardly obtained from (\ref{diarep})
and (\ref{Lop_bar}). Although their form for general values of $n$ and $L$ can be rather cumbersome
we still find compact expressions for some of them. For instance, the operator $\gen{\Omega}_{L}$ is trivial
in consonance with the fact that the leading term coefficient of the transfer matrix corresponds to the Cartan element
of the $U_q[\widehat{\alg{sl}}(2)]$ algebra \cite{Reshet_1987}. 
Fortunately, the situation is more interesting for the operator $\gen{\Omega}_{L-1}$ and we find a compact structure containing only
derivatives $\frac{\partial^{L-1}}{\partial x_i^{L-1}}$. In what follows we present the explicit differential equation
associated with the operator $\gen{\Omega}_{L-1}$ spectral problem.

\paragraph{The operator $\gen{\Omega}_{L-1}$.} Equation (\ref{eigK}) for $k=L-1$ and arbitrary values
of $n$ and $L$ corresponds to the following partial differential equation,
\[
\label{PDE}
\left[ \mathcal{V}^{(n)} + \sum_{i=1}^n  \mathcal{Q}^{(n)}_i \frac{\partial^{L-1}}{\partial x_i^{L-1}} \right] \bar{\mathcal{F}}_n (\bar{X}^{1,n} ) = \Delta_{L-1}  \bar{\mathcal{F}}_n ( \bar{X}^{1,n} ) \; ,
\]
with functions $\mathcal{V}^{(n)} = \mathcal{V}^{(n)}( \bar{X}^{1,n} )$ and $\mathcal{Q}^{(n)}_i = \mathcal{Q}^{(n)}(x_i ; \bar{X}_i^{1,n})$ 
defined over the sets $\bar{X}^{a,b}$ and $\bar{X}_i^{a,b} = \bar{X}^{a,b} \backslash \{ x_i \}$.
By writing $\mathcal{V}^{(n)} = - 2^{-L} \prod_{k=1}^L y_k^{-\frac{1}{2}} [\mathcal{V}_1^{(n)} + \mathcal{V}_2^{(n)}]$ we then have
\[
\label{Q1}
\mathcal{V}_1^{(n)} = (q^n + q^{L-n-2}) \sum_{k=1}^{L} y_k \; ,
\]
and
\<
\label{Q2}
\mathcal{V}_2^{(n)} = \begin{cases}
q^{n-2} (q-1)^2 (q+1) \sum_{k=0}^{L+1-2n} q^k \sum_{i=1}^{n} x_i \qquad \quad \quad \; L \geq 2(n-1) \nonumber \\
-q^{L-n} (q-1)^2 (q+1) \sum_{k=0}^{2n - 3 - L} q^k \sum_{i=1}^{n} x_i \qquad \quad \; L < 2(n-1) 
\end{cases} \; . \nonumber \\
\>
In their turn the functions $\mathcal{Q}^{(n)}_i$ are given by
\[
\label{VI}
\mathcal{Q}^{(n)}_i = \frac{(q-1)^2 (q+1)}{2^L q^{L+n} (L-1)!} \frac{\prod_{k=1}^{L} y_k^{-\frac{1}{2}}}{\prod_{\stackrel{j=1}{j \neq i}}^{n}(x_j - x_i)} \left[ \sum_{m=0}^{L} \mathcal{G}^{(n)}_m (x_i ; \bar{X}_i^{1,n}) \sum_{1 \leq j_1 < \dots < j_m \leq L} \prod_{\alpha=1}^{m} y_{j_{\alpha}} \right] \; , \nonumber \\
\]
where
\[
\label{GD}
\mathcal{G}^{(n)}_{L-d} (x_i ; \bar{X}_i^{1,n}) = x_i^d \sum_{l=0}^{n-1} x_i^l \; \psi_{l,d} \sum_{\stackrel{1 \leq j_1 < \dots < j_{n-1-l} \leq n}{j_{\alpha} \neq i}} \prod_{\alpha=1}^{n-1-l} x_{j_{\alpha}}
\]
and
\<
\label{PSI}
&& \psi_{l,d} = \nonumber \\
&& \begin{cases}
(-1)^{L+d+l} q^{L+2l} \sum_{k=0}^{2d + 2n - 3 - L - 4l} q^k \qquad \qquad \qquad \; \; d > L - (n+1) + 2l \nonumber \\
(-1)^{3l - n -1} q^{L+2l} \sum_{k=0}^{L-5} q^k \qquad \qquad \qquad \qquad \qquad  d = L - (n+1) + 2l , \; L \geq 5 \nonumber \\
(-1)^{3l - n} q^{2L + 2 l -4} \sum_{k=0}^{3-L} q^k \qquad \qquad \qquad \qquad \quad \; \;  d = L - (n+1) + 2l , \; L < 5 \nonumber \\
(-1)^{L+d+l+1} q^{2d+2n-2-2l} \sum_{k=0}^{L-2d-2n+1+4l} q^k \qquad \quad  d < L - (n+1) + 2l  
\end{cases} \; . \nonumber \\
\>
Clearly the terms $\sum_{k=0}^{l} q^k$ in (\ref{Q2}) and (\ref{PSI}) can be simplified with the help
of the geometric sum formula $\sum_{k=0}^{l} q^k = \frac{1-q^{l+1}}{1-q}$. However, we prefer to keep
the summation symbol in order to make more explicit that the summation vanishes for $l<0$. Eq. (\ref{PDE}) is a partial differential
equation of order $L-1$ and in what follows we shall demonstrate how it can be translated into a system of first order
equations. We shall also discuss its solutions and properties for particular values of $n$ and $L$.

%%%%%%%%%%%%%%%%%%%%%%%%%%%%%%%%%%%%%%%%%%%%%%%%%%%%%%%%%%%%%%%%%%%%%%%%%%%%%%%%
\subsection{Reduction of order}
\label{sec:RED}

Linear differential equations of higher order can be conveniently written as a system of 
first order equations. In our case we have (\ref{PDE}), which is a linear partial differential
equation of order $L-1$, and here we intend to embed that equation into a system of first order 
equations. The resulting system of partial differential equations is explicitly given in Lemma \ref{red}.

\begin{lemma} \label{red}
Let $\partial_i \equiv \frac{\partial}{\partial x_i}$ and let $\vec{\psi}$ be a $(L-2)n+1$ dimensional
vector denoted as
\[
\label{PSI}
\vec{\psi} = \left( \begin{matrix}
\psi^{(0)} \\ \psi^{(1)} \\ \vdots \\ \psi^{(L-2)} 
\end{matrix} \right) \; .
\]
Also define $\psi^{(0)} = \psi_0 = \bar{\mathcal{F}}_n (\bar{X}^{1,n} )$ while the remaining entries $\psi^{(k)}$
are $n$-dimensional column vectors with components $\psi_i^{(1)} = \partial_i \psi_0$ and
$\psi_i^{(k)} = \partial_i \psi_i^{(k-1)}$ for $k > 1$. The equation (\ref{PDE}) then reads
\[
\label{pdesys}
(\mathcal{Q}^{(n)} - \Delta_{L-1} ) \psi_0 +  \sum_{i=1}^n  \mathcal{V}^{(n)}_i \partial_i \psi_i^{(L-2)} = 0 \; .
\]
\end{lemma}
\begin{proof}
Straightforward substitution of (\ref{PSI}) into (\ref{PDE}).
\end{proof}

\paragraph{Matricial form.} The system of equations described in Lemma \ref{red} can be conveniently
written as a matrix equation. More precisely, the aforementioned system consists of the following equations
\<
\label{base}
\partial_i \psi_0 - \psi_i^{(1)} &=& 0 \nonumber \\
\partial_i \psi_i^{(k-1)} - \psi_i^{(k)} &=& 0 \qquad 1 < k \leq L-2 \; ,
\>
in addition to (\ref{pdesys}). Here we intend to rewrite (\ref{pdesys}) and (\ref{base})
as $\gen{\Upsilon} \vec{\psi} = 0$ for a given matrix $\gen{\Upsilon}$. For that we introduce the $n$-dimensional
vectors
\<
\label{omna0}
\vec{\omega}_0 = \left( \mathcal{Q}^{(n)}_1 \partial_1 , \dots , \mathcal{Q}^{(n)}_n \partial_n \right) 
\quad \mbox{and} \quad \vec{\nabla}_0 = \left( \begin{matrix} \partial_1 \\ \vdots \\ \partial_n \end{matrix} \right) \; .
\>
Also let $\mathbf{0}_{r \times s}$ denote the null matrix with dimensions $r \times s$ and define vectors
\<
\label{omna}
\vec{\omega} = \left( \mathbf{0}_{1 \times n(L-3)} , \vec{\omega}_0 \right) 
\quad \mbox{and} \quad \vec{\nabla} = \left( \begin{matrix} \vec{\nabla}_0 \\ \mathbf{0}_{n(L-3) \times 1} \end{matrix} \right) \; .
\>
Next we define $\hat{\mathfrak{D}}$ as a matrix of dimensions $(L-2)\times (L-2)$ with entries
\[
\label{DD}
\hat{\mathfrak{D}}_{ij} = \begin{cases}
- \mathbbm{1}_{n \times n} \qquad \quad i=j \;\;\; \quad ; \; 1 \leq j \leq L-2 \cr
\mathcal{D} \qquad \qquad \quad i=j+1 \; ; \; 1 \leq j < L-2 \cr
\mathbf{0}_{n \times n} \qquad \qquad \mbox{otherwise}  
\end{cases} \; ,
\]
where $\mathbbm{1}_{n \times n}$ denotes the $n \times n$ identify matrix and
$\mathcal{D}$ is also a  $n \times n$ diagonal matrix given by $\mathcal{D}= \mbox{diag}(\partial_1, \partial_2 , \dots , \partial_n)$.
In this way the system of Eqs. formed by (\ref{pdesys}) and (\ref{base}) can be written as
$\gen{\Upsilon} \vec{\psi} = 0$ with matrix $\gen{\Upsilon}$ given by
\[
\label{Ups}
\gen{\Upsilon} = \left( \begin{matrix}
\mathcal{Q}^{(n)} - \Delta_{L-1}  & \vec{\omega} \\
\vec{\nabla} & \hat{\mathfrak{D}} \end{matrix} \right) \; .
\]

As previously remarked in \Secref{sec:PDE}, the above discussed relation between functional equations of type 
(\ref{eigenL}) and partial differential equations (PDEs) was firstly proposed in \cite{Galleas11} for the partition 
function of the six-vertex model with domain wall boundaries. This relation was subsequently made precise in \cite{Galleas_proc}.
Despite their similarity, Eq. (\ref{eigenL}) describes an eigenvalue problem while the equation derived in \cite{Galleas_proc}
can not be regarded in that way. Nevertheless, in both cases the associated PDEs exhibit a similar structure, and one can wonder
if the equation derived in \cite{Galleas_proc} can also be recasted as a system of first order PDEs. 
In \Appref{sec:dwbc} we address this question and show that this is indeed the case for the PDE describing the partition function 
of the six-vertex model with domain wall boundaries.

%%%%%%%%%%%%%%%%%%%%%%%%%%%%%%%%%%%%%%%%%%%%%%%%%%%%%%%%%%%%%%%%%%%%%%%%%%%%%%%%
\subsection{Some particular solutions}
\label{sec:n0n1}

In this section we study the solutions of Eq. (\ref{PDE}) for the cases $n=0,1,2$
and particular values of the lattice length $L$. Interestingly, for the case $L=2$
and $n=2$ some geometric features of our equation emerge through the method of
characteristics \cite{Evans_book}.

%%%%%%%%%%%%%%%%%%%%%%%%%%%%%%%%%%%%%%%%%%%%%%%%%%%%%%%%%%%%%%%%%%%%%%%%%%%%%%%%
\subsubsection{Case $n=0$}
\label{sec:n0}

Although the case $n=0$ is trivial, we present it here for completeness reasons. 
By definition we have that $\bar{\mathcal{F}}_0$ is a constant and we can conclude
that $\Delta_{L-1} = \mathcal{V}^{(0)}$. Thus from (\ref{Q1}) and (\ref{Q2}) we find 
\[
\Delta_{L-1} = - \frac{(1+q^{L-2}) \sum_{k=1}^{L} y_k}{2^L \prod_{k=1}^{L} y_k^{\frac{1}{2}}} \; .
\]

%%%%%%%%%%%%%%%%%%%%%%%%%%%%%%%%%%%%%%%%%%%%%%%%%%%%%%%%%%%%%%%%%%%%%%%%%%%%%%%%
\subsubsection{Case $n=1$}
\label{sec:n1}

Equation (\ref{PDE}) for $n=1$ is actually an ordinary differential equation
reading
\[
\label{N1}
\frac{d^{L-1} \bar{\mathcal{F}}_1}{d x_1^{L-1}} = \left( \frac{\Delta_{L-1} - \mathcal{V}^{(1)}}{\mathcal{Q}^{(1)}_1} \right) \bar{\mathcal{F}}_1 \; .
\]
Moreover, when $L=2$ we can see that (\ref{N1}) is a first order equation and the general solution can be obtained by
direct integration. In that case we obtain
\<
\label{solN1}
\bar{\mathcal{F}}_1 (x_1) &=& \mathcal{C}_1 (q^2 x_1^2 - y_1 y_2)^{\frac{1}{2}} \exp{ \left\{ - \frac{q}{(q^2 -1)^2} \left[ (1+q^2)\frac{(y_1 + y_2)}{(y_1 y_2)^{\frac{1}{2}}} + 4 q \Delta_{L-1} \right] \xi (x_1)\right\} } \; , \nonumber \\
\>
where $\mathcal{C}_1$ is an integration constant and $\xi (x_1) = \arctanh{( q x_1 (y_1 y_2)^{-\frac{1}{2}})}$.
We are interested in solutions $\bar{\mathcal{F}}_1 \in \mathbb{K}^1 [x_1]$ while (\ref{solN1}) consists of a square root multiplied by an exponential. 
At first glance this structure does not resemble the desired class of solutions but we then notice that
\[
\exp{[ \xi (x_1) ]} = - \ii \frac{[q x_1 + (y_1 y_2)^{\frac{1}{2}}]}{(q^2 x_1^2 - y_1 y_2)^{\frac{1}{2}}} \; .
\]
In this way we find that the condition
\[
- \frac{q}{(q^2 -1)^2} \left[ (1+q^2)\frac{(y_1 + y_2)}{(y_1 y_2)^{\frac{1}{2}}} + 4 q \Delta_{L-1} \right] = \pm 1 
\]
leave us with the desired type of solution. Hence, the requirement $\bar{\mathcal{F}}_1 \in \mathbb{K}^1 [x_1]$
yields a constraint for the eigenvalues $\Delta_{L-1}$ which is solved by
\[
\Delta_{L-1} = - \frac{(1+q^2)}{4 q} \frac{(y_1 + y_2)}{(y_1 y_2)^{\frac{1}{2}}} \mp \frac{(q^2 - 1)^2}{4 q^2} \; .
\]

%%%%%%%%%%%%%%%%%%%%%%%%%%%%%%%%%%%%%%%%%%%%%%%%%%%%%%%%%%%%%%%%%%%%%%%%%%%%%%%%
\subsubsection{Case $n=2$}
\label{sec:char}
Here we shall address the case $n=2$ and $L=2$ where (\ref{PDE}) reads
\[
\label{pde2}
\mathcal{Q}^{(2)}_1 \frac{\partial \bar{\mathcal{F}}_2 }{\partial x_1} + \mathcal{Q}^{(2)}_2 \frac{\partial \bar{\mathcal{F}}_2 }{\partial x_2} = (\Delta_{L-1} - \mathcal{V}^{(2)}) \bar{\mathcal{F}}_2  \; .
\]
For this particular case it is worth remarking that $\mathcal{V}^{(2)}$ does not depend on the variables $x_1$ and $x_2$.
Now let $\mathcal{S} = \{ (x_1 , x_2 , \bar{\mathcal{F}}_2 ) \}$ be the surface generated by the solution of (\ref{pde2})
and let $\mathcal{C}$ be a curve lying on $\mathcal{S}$. Also, let $s$ be a variable parameterizing the curve $\mathcal{C}$
such that the vector
\[
\label{vf}
\left( \mathcal{Q}^{(2)}_1 (x_1(s) , x_2(s)) , \mathcal{Q}^{(2)}_2 (x_1(s) , x_2(s)) , (\Delta_{L-1} - \mathcal{V}^{(2)}) \bar{\mathcal{F}}_2 (x_1(s) , x_2(s)) \right)
\]
is tangent to $\mathcal{C}$ at each point of the curve. Then the curve 
\[
\mathcal{C} = \{ ( x_1 (s), x_2 (s) , \bar{\mathcal{F}}_2 (x_1(s) , x_2(s)) ) \}
\]
satisfy the following system of ordinary differential equations,
\<
\label{ODE}
\frac{d x_1}{ds} &=& \mathcal{Q}^{(2)}_1 (x_1(s) , x_2(s)) \nonumber \\
\frac{d x_2}{ds} &=& \mathcal{Q}^{(2)}_2 (x_1(s) , x_2(s)) \nonumber \\
\frac{d \bar{\mathcal{F}}_2}{ds} &=& (\Delta_{L-1} - \mathcal{V}^{(2)}) \bar{\mathcal{F}}_2 (x_1(s) , x_2(s))  \; .
\>
The curve $\mathcal{C}$ is called characteristic curve for the vector field (\ref{vf}) and it
is determined by the solution of the system (\ref{ODE}). The characteristic equations (\ref{ODE}) can also
be written without fixing a particular parameterization variable as
\[
\label{charod}
\frac{d x_1}{\mathcal{Q}^{(2)}_1} = \frac{d x_2}{\mathcal{Q}^{(2)}_2} = \frac{d \bar{\mathcal{F}}_2}{ (\Delta_{L-1} - \mathcal{V}^{(2)}) \bar{\mathcal{F}}_2} \; .
\]
Now we can form any two equations by combining the terms of (\ref{charod}) and for convenience we choose
\<
\label{charod1}
\frac{d x_1}{d x_2} &=& \frac{\mathcal{Q}^{(2)}_1}{\mathcal{Q}^{(2)}_2} = - \frac{x_1}{x_2} \nonumber \\
\frac{d \bar{\mathcal{F}}_2}{d x_2} &=&  (\Delta_{L-1} - \mathcal{V}^{(2)}) \frac{\bar{\mathcal{F}}_2}{\mathcal{V}^{(2)}_2} \; .
\>
The integration of (\ref{charod1}) yields the solution
\[
\label{solcur}
\bar{\mathcal{F}}_2 (x_1 , x_2) = \kappa(x_1 x_2) \exp{ \left\{  \frac{[(1+q^4)(y_1 + y_2) + 4q^2 \sqrt{y_1 y_2} \Delta_{L-1}]}{(q^2 -1)^2 (y_1 + y_2)} \log{\zeta(x_1 , x_2)} \right\} } \; ,
\]
where $\kappa(x_1 x_2)$ is an arbitrary function of the product $x_1 x_2$ and 
\[
\zeta(x_1 , x_2) = q^2 (y_1 + y_2) (x_1 + x_2) - (1+q^2)( y_1 y_2 + q^2 x_1 x_2 ) \; .
\]
Hence, in order to having $\bar{\mathcal{F}}_2 \in \mathbb{K}^1 [x_1 , x_2]$ we choose $\kappa$
as a constant function and impose the condition
\[
\label{cond}
\frac{[(1+q^4)(y_1 + y_2) + 4q^2 \sqrt{y_1 y_2} \Delta_{L-1}]}{(q^2 -1)^2 (y_1 + y_2)} = 1 \; .
\]
The resolution of (\ref{cond}) for $\Delta_{L-1}$ yields the eigenvalue
\[
\Delta_{L-1} = - \frac{(y_1 + y_2)}{2 \sqrt{y_1 y_2}} \; .
\]

%%%%%%%%%%%%%%%%%%%%%%%%%%%%%%%%%%%%%%%%%%%%%%%%%%%%%%%%%%%%%%%%%%%%%%%%%%%%%%%%
\section{Concluding remarks}
\label{sec:conclusion}

In this work we have presented a mechanism allowing to associate a linear partial differential
equation with the eigenvalue problem of the six-vertex model transfer matrix.
This mechanism has its roots in the algebraic-functional approach described in details in
\cite{Galleas_proc}, and a crucial step in this program is the identification of the
function space $\gen{\Xi} (\mathbb{C}^n)$ defined in \Secref{sec:Fn}. 

In \Secref{sec:OPD} we identify the action of the transfer matrix on the function space 
$\gen{\Xi} (\mathbb{C}^n)$ and we find that it is given in terms of certain operators $D_{z_i}^{z_{\alpha}}$
exhibiting a simple action even on larger spaces such as $\mathbb{C}[z]$ with $z = (z_1 , \dots , z_n) \in \mathbb{C}^n$.
Also, the operators $D_{z_i}^{z_{\alpha}}$ play a fundamental role in establishing the aforementioned
partial differential equations as they possess a differential realization in the desired function
space.

The six-vertex model transfer matrix admits the series expansion $T(\lambda) = T(0) (\mathbbm{1} + \mathcal{H} \lambda + \dots )$
where $T(0)$ is the discrete translation operator, i.e. momentum operator exponentiated, and $\mathcal{H}$ is the hamiltonian of the
$XXZ$ spin chain \cite{Korepin_book}. Although $\mathcal{H}$ contains only next-neighbors interactions, the higher order
terms contain highly non-local terms. On the other hand, our transfer matrix in the appropriate variable
consists essentially of a polynomial whose degree scales linearly with the lattice length. Thus the number of independent
commuting quantities having $T(\lambda)$ as its former are finite and, in particular, we find the operator 
$\gen{\Omega}_{L-1}$ which exhibits a local structure in terms of differentials.

The partial differential equation corresponding to the operator $\gen{\Omega}_{L-1}$ spectral problem
is explicitly given in \Secref{sec:PDE} and it consists of a linear equation of order $L-1$. Although 
we have an equation whose order scales with the lattice length $L$, in \Secref{sec:RED} we also
show how this equation can be translated into a system of first order partial differential equations
by standard methods. In \Secref{sec:n0n1} we analyze the solutions of our equation for particular values of the lattice
length. Interestingly, we find that the spectrum of eigenvalues is fixed by the condition that the eigenfunctions
belong to $\gen{\Xi} (\mathbb{C}^n)$.

%%%%%%%%%%%%%%%%%%%%%%%%%%%%%%%%%%%%%%%%%%%%%%%%%%%%%%%%%%%%%%%%%%%%%%%%%%%%%%%%
\section{Acknowledgements}
\label{sec:ack}
This work is supported by the Netherlands Organization for Scientific Research (NWO) under the VICI grant 680-47-602
and by the ERC Advanced grant research programme No. 246974, {\it ``Supersymmetry: a window to non-perturbative physics"}.
The work of W.G. is also supported by the German Science Foundation (DFG) under the Collaborative Research 
Center (SFB) 676 Particles, Strings and the Early Universe.

%%%%%%%%%%%%%%%%%%%%%%%%%%%%%%%%%%%%%%%%%%%%%%%%%%%%%%%%%%%%%%%%%%%%%%%%%%%%%%%%
%%%%%%%%%%%%%%%%%%%%%%%%%%%%%%%%%%%%%%%%%%%%%%%%%%%%%%%%%%%%%%%%%%%%%%%%%%%%%%%%
%%%%%%%%%%%%%%%%%%%%%%%%%%%%%%%%%%%%%%%%%%%%%%%%%%%%%%%%%%%%%%%%%%%%%%%%%%%%%%%%
\appendix

%%%%%%%%%%%%%%%%%%%%%%%%%%%%%%%%%%%%%%%%%%%%%%%%%%%%%%%%%%%%%%%%%%%%%%%%%%%%%%%%
\section{Domain wall boundaries}
\label{sec:dwbc}
The possibility of extracting partial differential equations from the algebraic-functional
method described in \Secref{sec:AF} has been first put forward in \cite{Galleas11} and
subsequently refined in \cite{Galleas_proc}. Those works have considered the six-vertex model
with domain wall boundary conditions and a compact partial differential equation describing
the model partition function has been derived in \cite{Galleas_proc}. This equation 
reads
\[
\label{om}
\left[ \sum_{i=1}^{L} \bar{a}(x_i , y_i) - \frac{1}{(L-1)!} \sum_{i=1}^{L} \prod_{j=1}^{L} \bar{a}(x_i , y_j) \prod_{\stackrel{j=1}{j \neq i}}^{L} \frac{\bar{a}(x_j , x_i)}{\bar{b}(x_j , x_i)} \frac{\partial^{L-1}}{\partial x_i^{L-1}} \right] \bar{Z} (\bar{X}^{1,L}) = 0 \; ,
\]
where $\bar{Z}$ is essentially the partition function of the model. In (\ref{om}) we are considering
the conventions $\bar{a}(x,y) = x q - y q^{-1}$ and $\bar{b}(x,y) = x  - y$ which is slightly different from the ones 
used in \cite{Galleas_proc}. Also, the structure of (\ref{om}) is closely related to the structure of (\ref{PDE}) and to make this feature
more apparent we rewrite (\ref{om}) as 
\[
\label{PDEdw}
\left[ \mathcal{V}^{DW} + \sum_{i=1}^L  \mathcal{Q}^{DW}_i \frac{\partial^{L-1}}{\partial x_i^{L-1}} \right] \bar{Z} (\bar{X}^{1,L} ) = 0 \; ,
\]
where 
\<
\label{VQ}
\mathcal{V}^{DW} &=& \sum_{i=1}^{L} \bar{a}(x_i , y_i) \nonumber \\
\mathcal{Q}^{DW}_i &=& - \frac{1}{(L-1)!} \prod_{j=1}^{L} \bar{a}(x_i , y_j) \prod_{\stackrel{j=1}{j \neq i}}^{L} \frac{\bar{a}(x_j , x_i)}{\bar{b}(x_j , x_i)} \; .
\>

Here we intend to translate (\ref{PDEdw}) into a system of first order partial differential equations
in the same lines of \Secref{sec:RED}. Since there are no significant modifications compared to
the reduction of order used for Eq. (\ref{PDE}), we restrict ourselves to presenting only the final result. 

\begin{lemma}
Let $\vec{\phi}$ be the following $L(L-2)+1$ dimensional vector
\[
\label{PHI}
\vec{\phi} = \left( \begin{matrix}
\phi^{(0)} \\ \phi^{(1)} \\ \vdots \\ \phi^{(L-2)} 
\end{matrix} \right) \; ,
\]
where $\phi^{(0)} = \phi_0 = \bar{Z} (\bar{X}^{1,L} )$ and the remaining entries are
$L$-dimensional column vectors with components $\phi_i^{(1)} = \partial_i \phi_0$ and
$\phi_i^{(k)} = \partial_i \phi_i^{(k-1)}$ for $k > 1$. Then Eq. (\ref{PDEdw}) is equivalent
to $\gen{\Upsilon}_{DW} \vec{\phi} = 0$ with $\gen{\Upsilon}_{DW}$ being obtained from
$\gen{\Upsilon}$ given in (\ref{Ups}) under the mappings $\Delta_{L-1} \mapsto 0$, 
$\mathcal{V}^{(n)} \mapsto \mathcal{V}^{DW}$, $\mathcal{Q}_i^{(n)} \mapsto \mathcal{Q}_i^{DW}$
and $n \mapsto L$.
\end{lemma}
\begin{proof}
Direct comparison of (\ref{PDE}), (\ref{PDEdw}) and (\ref{Ups}).
\end{proof}

%%%%%%%%%%%%%%%%%%%%%%%%%%%%%%%%%%%%%%%%%%%%%%%%%%%
%%%%%%%%%%%%%%%%%%%%%%%%%%%%%%%%%%%%%%%%
\bibliographystyle{hunsrt}
\bibliography{references}

\end{document}